\documentclass[final,copyright,creativecommons]{eptcs}
\providecommand{\event}{SR 2014} 

\usepackage{amsthm}
\usepackage{amsmath}
\usepackage{amssymb}

\newtheorem{lem}{Lemma}
\newtheorem{dfn}{Definition}
\newtheorem{thm}{Theorem}

\newtheorem{ex}{Example}

\newtheorem{cor}{Corollary}
\newtheorem{prop}{Proposition}

\def\Nat{\mathbb{N}}

\def\Xlength{{\textsf{cyc-EvenLen}}}
\def\Xparity{{\textsf{cyc-Parity}}}
\def\Xlarge{{\textsf{cyc-MaxFirst}}}
\def\Xmean{{\textsf{cyc-MeanPayoff}}}
\def\Xparityenergy{{\textsf{cyc-GoodForEnergy}}}
\def\Xenergy{{\textsf{cyc-Energy}}}

\def\Xlast{{\textsf{cyc-EndsZero}}}

\def\parity{\textsc{Parity}}
\def\energy{\textsc{Energy}}
\def\energyparity{\textsc{Energy-Parity}}

\def\universe{\mathbb{U}}

\newcommand{\A}{{\mathcal A}}
\newcommand{\FC}[1]{\textsc{FCG}({#1})}
\newcommand{\AC}[1]{\textsc{ACG}({#1})}
\newcommand{\EAC}[1]{\textsc{SCG}({#1})}

\title{First Cycle Games\thanks{This work is supported by the Austrian Science Fund through grant P23499-N23, and through the RiSE network (S11403-N23, S11407-N23); ERC Start grant (279307: Graph Games); and Vienna Science and Technology Fund (WWTF) grant PROSEED Nr. ICT 10-050.}
}
\author{Benjamin Aminof
\institute{IST Austria\\Vienna, Austria}
\email{benj@ist.ac.at}
\and
Sasha Rubin
\institute{IST Austria and TU Wien\\
Vienna, Austria}
\email{\quad srubin@ist.ac.at}
}
\def\titlerunning{First Cycle Games}
\def\authorrunning{B. Aminof \& S. Rubin}

\def\separable{unambiguous}
\def\separability{unambiguity}

\begin{document}

\maketitle

\begin{abstract}
First cycle games (FCG) are played on a finite graph by two players who push a token
along the edges until a vertex is repeated, and a simple cycle is formed. The winner is
determined by some fixed property $Y$ of the sequence of labels of the edges (or nodes)
forming this cycle. These games are traditionally of interest because of their
connection with infinite-duration games such as parity and mean-payoff games.

We study the memory requirements for winning strategies of FCGs and certain associated
infinite duration games. We exhibit a simple FCG that is not memoryless determined (this
corrects a mistake in {\it Memoryless determinacy of parity and mean payoff games: a
simple proof} by Bj\"orklund, Sandberg, Vorobyov (2004) that claims that FCGs for which
$Y$ is closed under cyclic permutations are memoryless determined). We show that
$\Theta(n)!$ memory (where $n$ is the number of nodes in the graph), which is always
sufficient, may be necessary to win some FCGs. On the other hand, we identify easy to
check conditions on $Y$ (i.e., $Y$ is closed under cyclic permutations, and both $Y$ and
its complement are closed under concatenation) that are sufficient to ensure that the
corresponding FCGs and their associated infinite duration games are memoryless
determined. We demonstrate that many games considered in the literature, such as
mean-payoff, parity, energy, etc., satisfy these conditions. On the complexity side, we
show  (for efficiently computable $Y$) that while solving FCGs is in PSPACE, solving
some families of FCGs is PSPACE-hard.
\end{abstract}

\section{Introduction}
First cycle games (FCGs) are played on a finite graph by two players who push a token
along the edges of the graph until a simple cycle is formed.
Player 0 wins the play if the sequence of labels of the edges (or nodes) of the cycle
satisfies some fixed cycle property $Y$, and otherwise Player 1 wins. For instance, if
every vertex has an integer priority, the cycle property $Y = \Xparity$ states that the
largest priority occurring on the cycle should be even. For a fixed cycle property $Y$,
we write $\FC{Y}$ for the family of games over all possible arenas with this winning
condition.
We are motivated by two questions: Under what conditions on $Y$ is every game in
$\FC{Y}$ memoryless determined? What is the connection between FCGs and
infinite-duration games?


\noindent {\bf  First cycle games.} First, we give a simple example showing that first
cycle games (FCGs) are not necessarily memoryless determined, even if $Y$ is closed
under cyclic permutations (i.e., even if winning depends on the cycle but not on how it
was traversed), contrary to the claim in \cite{BSV}[Page $370$]. We then show that, for
a graph with $n$ nodes, whereas no winning strategy needs more than $(n-1)!$ memory
(since this is enough to remember the whole history of the game), some FCGs require at
least $\Omega(n!)$ memory. To complete the picture, we analyse the complexity of solving
FCGs and show that it is PSPACE-complete. More specifically, we show that if one can
decide in PSPACE whether a given cycle satisfies the property $Y$, then solving the
games in $\FC{Y}$ is in PSPACE; and that even for a trivially computable cycle property
$Y$ (namely, that the cycle ends with the label $0$), solving the games in $\FC{Y}$ is
PSPACE-hard.

\noindent {\bf  First Cycle Games and Infinite-Duration Games.} The main object used to
connect FCGs and infinite-duration games (such as parity games) is the
\emph{cycles-decomposition} of a path. Informally, a path is decomposed by pushing the
edges of the path onto a stack; as soon as a cycle is detected in the stack it is popped
and output, and the algorithm continues. We then say that a winning condition $W$ (such
as the parity or energy winning condition) is {\em $Y$-greedy on $\A$} if in the game on
arena $\A$ and winning condition $W$, Player $0$ is guaranteed to win if he ensures that
every cycle in the cycles-decomposition of the play satisfies $Y$, and Player $1$ is
guaranteed to win if she ensures that every cycle in the cycles-decomposition does not
satisfy $Y$. We prove a {\em Transfer Theorem:} if $W$ is $Y$-greedy on $\A$, then the
winning regions in the following two games on arena $\A$ coincide, and memoryless
winning strategies transfer between them: the infinite duration game with winning
condition $W$, and the FCG with winning condition $Y$.

To illustrate the usefulness of the concept of being $Y$-greedy, we instantiate the
definition to well-studied infinite-duration games: i) the parity winning condition (the
largest priority occurring infinitely often is even) is $Y$-greedy on every arena $\A$
where $Y = \Xparity$, ii) the mean-payoff condition (the mean payoff is at least $\nu$)
is $\Xmean_\nu$-greedy on every arena $\A$ (where $\Xmean_\nu = $ average payoff is at
least $\nu$), and iii) for every arena $\A$ with vertex set $V$, and largest weight $W$,
the energy condition stating that the energy level is always non-negative starting with
initial credit $W (|V|-1)$  is $\Xenergy$-greedy on $\A$ (where $\Xenergy = $ the energy
level is non-negative).

In order to prove memoryless determinacy of certain FCGs (and related infinite-duration
games) we generalise techniques used to prove that mean-payoff games are memoryless
determined (Ehrenfeucht and Mycielski \cite{EM79}). Given a cycle property $Y$, we first
consider the infinite duration games {\em $\AC{Y}$} (all cycles), and {\em  $\EAC{Y}$}
(suffix all-cycles). A game in the family $\AC{Y}$ requires Player $0$ to ensure that
every cycle in the cycles-decomposition of the play (starting from the beginning)
satisfies $Y$. A game in the family $\EAC{Y}$ requires Player $0$ to ensure that every
cycle in the cycles-decomposition of \emph{some suffix} of the play satisfies $Y$. As
was done in~\cite{EM79}, reasoning about infinite and finite duration games is
intertwined -- in our case, we simultaneously reason about games in $\FC{Y}$ and
$\EAC{Y}$. We define a property of arenas, which we call {\em $Y$-\separable}, and prove
a {\em Memoryless Determinacy Theorem}: a game from $\FC{Y}$ whose arena $\A$ is
$Y$-\separable\ is  memoryless determined. Combining this with the Transfer Theorem
above, we also get that if $\A$ is $Y$-\separable, then any game with a winning
condition $W$ that is $Y$-greedy on $\A$, is memoryless determined\footnote{Taking $Y$
to be $\Xparityenergy$ (defined to be that either the energy level is positive, or it is
zero and the largest priority occurring is even) and noting that for every arena $\A$ we
have: i) $\A$ is $Y$-\separable\ and, ii) the game in $\AC{Y}$ over $\A$ is $Y$-greedy
on $\A$; we obtain a proof of \cite{ChDoEnergyParity}[Lemma $4$] that no longer relies
on the incorrect result from \cite{BSV}.}.

Although checking if an arena is $Y$-\separable\ may not be hard, it has two
disadvantages: it involves reasoning about infinite paths and it involves reasoning
about the arena whereas, in many cases, memoryless determinacy is guaranteed by the
cycle property $Y$ regardless of the arena (this is the case for example with $Y =
\Xparity$). Therefore, we also provide easy to check `finitary' sufficient conditions on
$Y$ (namely that $Y$ is closed under cyclic permutations, and both $Y$ and its
complement are closed under concatenation) that ensure $Y$-\separability\ of every
arena, and thus memoryless determinacy for all games in $\FC{Y}$. We demonstrate the
usefulness of these conditions by observing that typical cycle properties are easily
seen to satisfy them, e.g., $\Xparity, \Xmean_\nu, \Xenergy$.

We conclude by noting that, in particular, if $Y$ is closed under cyclic permutations,
and both $Y$ and its complement are closed under concatenation, then games with winning
condition $W$ are memoryless determined on every arena $\A$ for which $W$ is $Y$-greedy
on $\A$. As noted above, for many
 winning conditions $W$ (such as mean-payoff, parity, and energy winning conditions) it is easy to
find a cycle property $Y$ satisfying the mentioned closure conditions, and for which $W$
is $Y$-greedy on the arena of interest. This provides an easy way to deduce memoryless
determinacy of these classic games.

\noindent {\bf  Related work.} As just discussed, this work extends~\cite{EM79}, finds a
counter-example to a claim in~\cite{BSV}, and supplies a proof of a lemma
in~\cite{ChDoEnergyParity}. Conditions that ensure (or characterise) which games have
memoryless strategies appear for example in \cite{BFFM11, GiZe05, Ko06}. However, all of
these deal with infinite duration games and do not exploit the connection to finite
duration games.

Due to space limitations, proofs appear in the full version of the article.

\section{Definitions}

In this paper all games are two-player turn-based games of perfect information played on
finite graphs. The players are called Player $0$ and Player $1$.

%

\noindent{\bf Arena} An {\em arena} is a labeled directed graph $\A =
(V_0,V_1,E,\universe, \lambda)$ where
\begin{enumerate}
\item $V_0$ and $V_1$ are disjoint sets of vertices of Player 0 and Player 1,
    respectively; the set of vertices of the arena $V := V_0 \cup V_1$ is non-empty.
\item $E \subseteq V \times V$ is a set of edges with no dead-ends (i.e., for every
    $v \in V$ there is some edge $(v,w) \in E$);
\item $\universe$ is a set of possible labels.
\item $\lambda:E \to \universe$ is a {\em labeling} function, used by the winning
    condition.
\end{enumerate}


Typical choices for $\universe$ are $\mathbb{R}$ and $\mathbb{N}$. Games in which
vertices are labeled instead of edges can be modeled by ensuring $\lambda(v,w) =
\lambda(v,w')$ for all $v,w,w' \in V$. Similarly, games in which vertices are labeled by
elements of $\universe'$ and edges are labeled by elements of $\universe''$ can be
modeled by labeling edges by elements of $\universe' \times \universe''$. As usual, if
$u = e_1 e_2 \cdots$ is a (finite or infinite) sequence of edges in the arena, we write
$\lambda(u)$ for the string of labels $\lambda(e_1) \lambda(e_2) \cdots$.


\noindent{\bf Plays and strategies}
	A {\em play} $\pi = \pi_0, \pi_1, \ldots$ in an arena is an infinite\footnote{For
simplicity, we consider plays of both finite and infinite duration games to be infinite.
However, in a finite duration game (and thus in any FCG) the winner is determined by a
finite prefix of the play, and the moves after this prefix are immaterial.} sequence
over $V$ such that $(\pi_j,\pi_{j+1}) \in E$ for all $j \in \mathbb{N}$. The node
$\pi_0$ is called the \emph{starting} node of the play. We denote the set of all plays
in the arena $\A$ by $plays(\A)$.
	A {\em strategy} for Player $i$ is a function $S:V^*V_i \to V$ such that if $u \in
V^*$ and $v \in V_i$ then $(v,S(uv)) \in E$.
	A strategy $S$ for Player $i$ is {\em memoryless} if $S(uv) = S(u'v)$ for all $u,u'
\in V^*, v \in V_i$.
	A play $\pi$ is {\em consistent} with $S$, where $S$ is a strategy for Player $i$,
if for every $j \in \Nat$  such that $\pi_j \in V_i$, it is the case that $\pi_{j+1} =
S(\pi_0 \cdots \pi_j)$. 	
	A strategy $S$ for Player $i$ is {\em generated by a Moore machine} if there exists
a finite set $M$ of {\em memory states}, an {\em initial state} $m_I \in M$, a {\em
memory update} function $\delta:V \times M \to M$, and a {\em next-move function}
$\rho:V \times M \to V$ such that if $u = u_0 u_1 \cdots u_l$ is a prefix of a play with
$u_l \in V_i$ then $S(u) = \rho(u_l,m_l)$ where $m_l$ is defined inductively by $m_0 =
m_I$ and $m_{i+1} = \delta(u_i,m_i)$. A strategy $S$ is {\em finite-memory} if it is
generated by some Moore machine. A strategy $S$ {\em uses memory at most $k$} if it is
generated by some Moore machine with $|M| \leq k$. A strategy $S$ {\em uses memory at
least $k$} if every Moore machine generating $S$ has $|M| \geq k$.

\noindent{\bf Games, Winning Conditions, and Memoryless Determinacy} A {\em game} is a
pair $(\A,O)$ where $\A = (V_0,V_1,E,\universe,\lambda)$ is an arena and $O \subseteq
plays(\A)$ is an {\em objective} (usually induced by the labeling). If either $V_0$ or
$V_1$ is empty, then the game $(A,O)$ is called a {\em solitaire game}. A play $\pi$ in
a game $(\A,O)$ is {\em won by Player $0$} if $\pi \in O$, and {\em won by Player $1$}
otherwise. A strategy $S$ for Player $i$ is {\em winning starting from a node $v \in V$}
if every play $\pi$ that starts from $v$ and is consistent with $S$ is won by Player
$i$.

 A {\em winning condition} is a set $W \subseteq \universe^\omega$.
If $W$ is a winning condition and $\A$ is an arena, the objective $O_W(\A)$ induced by
$W$ is defined as follows: $O_W(\A) = \{ v_0 v_1 v_2 \cdots \in plays(\A) \mid
\lambda(v_0,v_1) \lambda(v_1,v_2) \cdots \in W \}$.
Here are some standard winning conditions:


  \begin{itemize}
  \item The {\em parity condition} $\parity$ consists of those infinite sequences
      $c_1 c_2 \cdots \in \mathbb{N}^\omega$ such that the largest label occurring
      infinitely often is even.
        \item  For $\nu \in \mathbb{R}$, the {\em $\nu$-mean-payoff condition}
            consists of those infinite sequences $c_1 c_2 \cdots \in  \mathbb{R}$
            such that $\lim\sup_{k \to \infty} \frac{1}{k} \sum_{i=1}^{k} c_i$ is at
            least $\nu$.
  \item The {\em energy condition} for a given {\em initial credit} $r \in \Nat$,
      written $\energy(r)$,  consists of those infinite sequences $c_1 c_2 \cdots
      \in \mathbb{Z}^\omega$ such that $r + c_1 + \cdots + c_k  \geq 0$ for all $k
      \geq 1$.
  \item The {\em energy-parity condition} $\energyparity(r)$ is defined as
      consisting of $(c_1,d_1) (c_2,d_2) \cdots \in  \mathbb{N} \times \mathbb{Z}$
      such that $c_1 c_2 \cdots$ is in $\parity$ and $d_1 d_2 \cdots$ is in
      $\energy(r)$.
  \end{itemize}

The {\em (memoryless) winning region} of Player $i$ is the set of vertices $v \in V$
such that Player $i$ has a (memoryless) winning strategy starting from $v$. A game is
{\em pointwise memoryless for Player $i$} if the memoryless winning region for Player
$i$ coincides with the winning region for Player $i$. A game is {\em uniform memoryless
for Player $i$} if there is a memoryless strategy for Player $i$ that is winning
starting from every vertex in that player's winning region.

 A game is {\em determined} if the winning regions partition $V$.
A game is {\em pointwise memoryless determined} if it is determined and it is pointwise
memoryless for both players. A game is {\em uniform memoryless determined} if it is
determined and uniform memoryless for both players.

\noindent{\bf Cycles-decomposition} A {\em cycle} in an arena $\A$ is a sequence of
edges $(v_1,v_2) (v_2,v_3) \cdots (v_{k-1},v_k)(v_k,v_1)$.

Define an algorithm that processes a play $\pi \in plays(\A)$ and outputs a sequence of
cycles: at step $0$ start with empty stack; at step $j$ push the edge
$(\pi_j,\pi_{j+1})$, and if for some $k$, the top $k$ edges on the stack form a cycle,
this cycle is popped and output, and the algorithm continues to step $j+1$. The sequence
of cycles output by this algorithm is called the {\em cycles-decomposition of $\pi$},
and is denoted by $cycles(\pi)$. The {\em first cycle of $\pi$} is the first cycle in
$cycles(\pi)$. For example, if $\pi = v w x w v s (x y z)^\omega$, then $cycles(\pi) =
(w,x)(x,w), (v,w)(w,v), (x,y)(y,z)(z,x), (x,y)(y,z)(z,x),
\ldots$, and the first cycle of $\pi$ is $(w,x)(x,w)$.
%
Note that $cycles(\pi)$ is such that at most $|V|-1$ edges of $\pi$ do not appear in it
(i.e, they are pushed but never popped -- like the edge $(v,s)$ in the example above).
As we show in the full version, this allows one to reason, for instance, about the
initial credit problem for energy games (cf. \cite{ChDoEnergyParity}).

\noindent{\bf Cycle properties} A {\em cycle property} is a set $Y \subseteq
\universe^*$, used later on to define winning conditions for games. Here are some cycle
properties that we refer to in the rest of the article:

\begin{enumerate}
\item Let $\Xlength$ be those sequences $c_1 c_2 \cdots c_k \in \universe^*$ such
    that $k$ is even.

\item Let $\Xparity$ be those sequences $c_1 \cdots c_k \in \mathbb{N}^*$ such that
    $\max_{1 \leq i \leq k} c_i$ is even.

\item Let $\Xenergy$ be those sequences $c_1 \cdots c_k \in \mathbb{Z}^*$ such that
    $\sum_{i=1}^k c_i \geq 0$.

\item Let $\Xparityenergy$ be those sequences $(c_1,d_1) \cdots (c_k,d_k) \in
    (\mathbb{N} \times \mathbb{Z})^*$ such that either $\sum_{i=1}^k d_i > 0$, or
    both $\sum_{i=1}^k d_i = 0$ and $c_1 \cdots c_k  \in \Xparity$.

\item Let $\Xmean_\nu$ be those sequences $c_1 \cdots c_k \in \mathbb{R}^*$ such
    that $\frac{1}{k} \sum_{i=1}^k c_i \leq \nu$,
    for some $\nu \in \mathbb{R}$. 

\item Let $\Xlarge$ be those sequences $c_1 \cdots c_k \in \mathbb{N}^*$ such that
    $c_1 \geq c_i$ for all $1 \leq i \leq k$.

\item Let $\Xlast$ be those sequences $c_1 \cdots c_k \in \mathbb{N}^*$ such that
    $c_k = 0$.

\end{enumerate}
If $Y \subseteq \universe^*$ is a cycle property, write $\neg Y$ for the cycle property
$\universe^* \setminus Y$.
We isolate two important classes of cycle properties (the first is inspired by \cite{BSV}):

\begin{enumerate}
\item Say that $Y$ is {\em closed under cyclic permutations} if $a b \in Y$ implies
    $ba \in Y$, for all $a \in \universe, b \in \universe^*$.
\item Say that $Y$ is {\em closed under concatenation} if $a \in Y$ and $b \in Y$
    imply that $ab \in Y$, for all $a,b \in \universe^*$.
\end{enumerate}

Note that the cycle properties 1-5 above are closed under cyclic permutations and
concatenation; and that $\neg \Xlength$ is closed under cyclic permutations but not
under concatenation.

\noindent{\bf First Cycle Games (FCGs)}
Given a cycle property $Y \subseteq \universe^*$, and an arena $\A =
(V_0,V_1,E,\universe, \lambda)$, let the objective $O_{\FC{Y}}(\A) \subseteq plays(\A)$
be such that $\pi \in O_{\FC{Y}}(\A)$ iff $\lambda(u) \in Y$ where $u$ is the {\em first
cycle} in the cycles-decomposition of $\pi$. The family $\FC{Y}$ of {\em first cycle
games of $Y$} consists of all games of the form $(A,O_{\FC{Y}}(\A))$ where $\A$ is an
arena with labels in $\universe$. For instance, $\FC{\Xparity}$ consists of those games
such that Player $0$ wins iff the largest label occurring on the first cycle is
even.\footnote{Formally, then, first cycle games are of infinite duration, although the
winner is determined after the first cycle appears on the play.}

\section{Finite Duration Cycle Games (on being first)}

In this section we analyse the memory required for winning strategies in first cycle
games, and the complexity of solving these games. We begin by correcting a mistake in
\cite{BSV}.

\begin{prop}
There exists a cycle property $Y$ closed under cyclic permutations and a game in
$\FC{Y}$ that is not
    pointwise memoryless determined.
\end{prop}

To see this, consider a game where Player 1 chooses from $\{a,b\}$ and Player 0 must
match the choice. This clearly requires Player 0 to have memory. The claim follows by
simply encoding this game as a FCG. For example, let the cycle-property $Y$ be
$\Xlength$, let the vertex set be $\{v_1,v_2,v_3,v_4\}$, let $V_0 = \{v_1\}$, and let
the edges be $\{(v_1,v_2), (v_2,v_1), (v_1,v_3), (v_3,v_2), (v_2,v_4),(v_4,v_1)\}$.

We now consider the difference between pointwise and uniform memoryless determinacy of
FCGs.

\begin{thm} \label{uniform}
\begin{enumerate}
\item Solitaire FCGs are pointwise memoryless determined.
\item There is a solitaire FCG that is not uniform memoryless determined.

\item \label{unif} If cycle property $Y$ is closed under cyclic permutations, and a
    game from $\FC{Y}$ is pointwise memoryless for Player $i$, then that game is
    uniform memoryless for Player $i$.
\end{enumerate}

\end{thm}

\begin{prop} \label{memory_bounds}
\begin{enumerate}
\item For a FCG on an arena with $n$ vertices, if Player $i$ wins from $v$, then
    every winning strategy for Player $i$ starting from $v$ uses memory at most
    $(n-1)!$.
\item For every $n$ there exists a FCG on an arena with $3n+1$ vertices, and a
    vertex $v$, such that every winning strategy for Player $0$ starting from $v$
    uses memory at least $n!$.
\end{enumerate}
\end{prop}

The first item is immediate since $(n-1)!$ is enough to remember the whole history of
the game up to the point a cycle is formed. The proof of the second item is by showing a
game where Player 1 can ``weave'' any possible permutation of $n$ nodes, whereas in
order to win Player 0 must remember this permutation. The construction is in the full
version of the paper.

Finally, we analyse the complexity of solving FCGs with efficiently computable cycle
properties.
\begin{thm} \label{thm:pspace}
\begin{enumerate}
\item If $Y$ is a cycle property for which solving membership is in $PSPACE$, then
    the problem of solving games in $\FC{Y}$ is in PSPACE.
\item The problem of solving games in $\FC{\Xlast}$ is PSPACE-complete.
\end{enumerate}
\end{thm}

\begin{proof}[Sketch]

For the first item,
observe that solving the game amounts to evaluating the finite AND-OR tree obtained by
unwinding the arena into all possible plays, up to the point on each play where a cycle
is formed; nodes belonging to Player $0$ are 'or' nodes, nodes belonging to Player $1$
are 'and' nodes, and a leaf is marked by 'true' iff the cycle formed on the way to it is
in $Y$. Since this tree has depth at most $n$ (the size of the arena), and since we
assumed membership in $Y$ is in PSPACE, marking the leaves can be done in PSPACE. So
evaluating the tree can be done in PSPACE.

For the second item, note that Generalised Geography can be thought of as a first cycle
game in which Player $i$ nodes are labeled by $i$, and $Y = \Xlast$. Note that computing
$Y$ is computationally trivial, but solving Generalised Geography is PSPACE-hard (see
for instance \cite{sipser}[Theorem $8.11$]).
\end{proof}

\section{Infinite Duration Cycle Games}

\subsection{On being greedy}

 We start by defining two types of infinite duration games
called the \emph{All-Cycles} and the \emph{Suffix All-Cycles} games, whose winning
condition is derived from $Y$. Informally, All-Cycles games are games in which Player
$0$ wins iff all cycles in the cycles-decomposition of the play are in $Y$, and Suffix
All-Cycles Games are games in which Player $0$ wins iff all cycles in the
cycles-decomposition of \emph{some suffix} of the play are in $Y$. Formally, for arena
$\A = (V_0,V_1,E,\universe,\lambda)$ and cycle property $Y \subseteq \universe^*$, we
define two objectives $O \subseteq plays(\A)$ and corresponding families of games as
follows:
\begin{enumerate}
\item $\pi \in O_{\AC{Y}}(\A)$ :if $\lambda(u) \in Y$ for {\em all cycles} $u$ in
    $cycles(\pi)$.
\item $\pi \in O_{\EAC{Y}}(\A)$ :if some {\em suffix} $\pi'$ of $\pi$ satisfies that
    $\lambda(u) \in Y$ for all cycles $u$ in $cycles(\pi')$.~\footnote{Note that
    this is \emph{not} the same as saying that $\lambda(u) \in Y$ for all but
    finitely many cycles $u$ in $cycles(\pi)$. For instance, let $Y$ be the property
    that the cycle has odd length, and take $ \pi := (v_1v_2v_1v_3v_2v_4)^\omega$.
    Note that i) decomposing the suffix $\pi'$ starting with the second vertex
    results in all cycles having odd length, and ii) it is not the case that almost
    all cycles in the cycles-decomposition of $\pi$ have odd length (in fact, they
    all have even length).}
\end{enumerate}

Define the corresponding families of games:
\begin{enumerate}
\item The family $\AC{Y}$ of {\em all-cycles games of $Y$} consists of all games of
    the form $(\A,O_{\AC{Y}}(\A))$.

\item The family $\EAC{Y}$ of {\em suffix all-cycles games of $Y$} consists of all
    games of the form $(\A,O_{\EAC{Y}}(\A))$.
\end{enumerate}

\begin{dfn}
Say that a {\em game $(\A,O)$ is $Y$-greedy} if $O_{\AC{Y}}(\A) \subseteq O$ and
$O_{\AC{\neg Y}}(\A) \subseteq V^\omega \setminus O$. Say that a {\em winning condition
$W$ is $Y$-greedy on arena $\A$} if the game $(\A,O_W)$ is $Y$-greedy.
\end{dfn}

Intuitively, $W$ being $Y$-greedy on $\A$ means that Player $0$ can win the game on
arena $\A$ with winning condition $W$ if he ensures that every cycle in the
cycles-decomposition of the play is in $Y$, and Player $1$ can win if she ensures that
every cycle in the cycles-decomposition of the play is not in $Y$.

For instance,  the winning condition $\parity$ (the largest priority occurring
infinitely often is even) is $\Xparity$-greedy on every arena $\A$, the
$\nu$-mean-payoff condition (the $\limsup$ average is at least $\nu$) is
$\Xmean_\nu$-greedy on every arena $\A$, and the energy condition (stating that the
energy level is always non-negative starting with initial credit $W (|V|-1)$, where $W$
is the largest weight and $V$ are the vertices of the arena $\A$) is $\Xenergy$-greedy
on $\A$.

\begin{thm}[Transfer] \label{greedy}
Let $(\A,O)$ be a $Y$-greedy game, and let $i \in \{0,1\}$.
\begin{enumerate}
\item The winning regions for Player $i$ in the games $(\A,O)$ and
    $(\A,O_{\FC{Y}}(\A))$ coincide.
\item For every memoryless strategy $S$ for Player $i$ starting from $v$ in arena
    $\A$: $S$ is winning in the game $(\A,O)$ if and only if $S$ is winning in the
    game $(\A,O_{\FC{Y}}(\A))$.
\end{enumerate}
\end{thm}

\begin{cor}
Let $W$ be $Y$-greedy on arena $\A$. Then the game $(\A,O_W)$ is determined, and is
pointwise (uniform) memoryless determined if and only if the game $(\A,O_{\FC{Y}}(\A))$
is pointwise (uniform) memoryless determined.
\end{cor}

\subsection{On being \separable}

\begin{dfn}
An arena $\A$ is  {\em $Y$-\separable} if $O_{\EAC{Y}}(\A) \cap O_{\EAC{\neg Y}}(\A) =
\emptyset$.
\end{dfn}

\begin{lem} \label{lem:EAgreedy}
If $\A$ is $Y$-\separable\ then the game $(\A,O_{\EAC{Y}}(\A))$ is $Y$-greedy.
\end{lem}

\begin{thm}[Memoryless Determinacy]  \label{EM}
If arena $\A$ is $Y$-\separable, then the game $(\A,O_{\FC{Y}}(\A))$ is pointwise
memoryless determined. If $Y$ is also closed under cyclic permutations, then this game
is uniform memoryless determined.
\end{thm}

It is of interest to note that the proof of this theorem is a generalisation of the
 proof used in~\cite{EM79} for showing memoryless determinacy of mean-payoff games. As
in~\cite{EM79}, our proof reasons about infinite plays. More specifically, we obtain
from Theorem~\ref{greedy} and Lemma~\ref{lem:EAgreedy} that the winning regions of each
player in the games $(\A,O_{\EAC{Y}}(\A))$ and $(\A,O_{\FC{Y}}(\A))$ coincide, and then
go on and use this fact to derive memoryless strategies for the game
$(\A,O_{\FC{Y}}(\A))$.

\begin{cor}\label{cor:greedy-separable}
Suppose arena $\A$ is $Y$-\separable.
\begin{enumerate}
\item If  $(\A,O)$ is $Y$-greedy, then the game $(\A,O)$ is pointwise memoryless
    determined.
\item The games $(\A,O_{\EAC{Y}}(\A))$ and $(\A,O_{\AC{Y}}(\A))$ are pointwise
    memoryless determined.
\end{enumerate}
If in addition $Y$ is closed under cyclic permutations, then these game are uniform
memoryless determined.
\end{cor}

\begin{proof}
For the first item combine Theorems~\ref{greedy} and \ref{EM}. For the second, use
Lemma~\ref{lem:EAgreedy} and the fact that $(\A,O_{\AC{Y}}(\A))$ is always $Y$-greedy.
For the final statement apply Theorem~\ref{uniform} item \ref{unif}.
\end{proof}

We now provide a simple sufficient condition on $Y$ --- that does not involve reasoning
about cycles-decompositions of infinite paths ---  that ensures that every arena $\A$ is
$Y$-\separable:

\begin{thm} \label{char}
Let $Y \subseteq \universe^*$ be a cycle property. If $Y$ is closed under cyclic
permutations\footnote{It may be worth noting that $Y$ is closed under cyclic
permutations iff so is $\neg Y$.}, and both $Y$ and $\neg Y$ are closed under
concatenation, then every arena $\A$ is $Y$-\separable.
\end{thm}

It is easy to check that the following cycle properties satisfy the hypothesis of
Theorem~\ref{char}: $\Xparity$, $\Xenergy$, $\Xmean_\nu$, and $\Xparityenergy$. On the
other hand, $\neg \Xlength$ is not closed under concatenation, whereas $\Xlarge$ is not
closed under cyclic permutations.

We conclude with the main result of this section:

\begin{cor} \label{cor:closed}
Suppose $Y$ is closed under cyclic permutations, and both $Y$ and its complement are
closed under concatenation. Then the following games are uniform memoryless determined
for every arena $\A$: $(\A,O_W)$ if $W$ is $Y$-greedy on $\A$,  $(\A,O_{\EAC{Y}}(\A))$,
and $(\A,O_{\AC{Y}}(\A))$.
\end{cor}

We believe that Corollary~\ref{cor:closed} provides a practical and easy way of deducing
that many infinite duration games are uniform memoryless determined, as follows: exhibit
a cycle property $Y$ that is closed under cyclic permutations and both $Y$ and $\neg Y$
are closed under concatenation, such that the winning condition $W$ is $Y$-greedy on the
arena $A$ of interest. Finding such a $Y$ is usually easy since it is simply a
`finitary' version of the winning condition $W$. For example, uniform memoryless
determinacy of parity games, mean-payoff games, and energy-games, can easily be deduced
by considering the cycle properties $\Xparity$, $\Xmean_\nu$, and $\Xenergy$.

\bibliographystyle{eptcs}
\bibliography{cyclegames}

\end{document}